\documentclass[12pt]{article}

\def\hybrid{\topmargin -20pt  \oddsidemargin 0pt
      \headheight 0pt   \headsep 0pt
      \textwidth 6.25in 
      \textheight 9.5in 
      \marginparwidth .875in
      \parskip 5pt plus 1pt   \jot = 1.5ex}

\hybrid

\usepackage[latin1]{inputenc}
\usepackage[T1]{fontenc}
\usepackage{amsmath}
\usepackage{amsthm}
\usepackage{amssymb}
\usepackage[all]{xy}
\usepackage{amscd,amsmath,amssymb,euscript}

\newtheorem{thm}{Theorem}[section]

\numberwithin{equation}{section}

\theoremstyle{definition}

\newtheorem{ex}[thm]{Example}

\renewcommand{\to}{\xymatrix@1@=15pt{\ar[r]&}}
\renewcommand{\rightarrow}{\xymatrix@1@=15pt{\ar[r]&}}
\renewcommand{\mapsto}{\xymatrix@1@=15pt{\ar@{|->}[r]&}}
\renewcommand{\twoheadrightarrow}{\xymatrix@1@=15pt{\ar@{->>}[r]&}}
\renewcommand{\hookrightarrow}{\xymatrix@1@=15pt{\ar@{^(->}[r]&}}
\newcommand{\congpf}{\xymatrix@1@=15pt{\ar[r]^-\sim&}}

\DeclareMathOperator{\tr}{tr}

\def\dbar{\bar\partial}

\begin{document}
\thispagestyle{empty}
\rightline{}
\vspace{2truecm}
\centerline{\bf  \Large Heterotic Non-K\"ahler Geometries via}
\vspace{.3truecm}
\centerline{\bf \Large Polystable Bundles on Calabi-Yau Threefolds}

\vspace{1.5truecm}
\centerline{Bj\"orn Andreas\footnote{andreas@math.hu-berlin.de}
and Mario Garcia-Fernandez\footnote{mariogf@qgm.au.dk, Supported by QGM (Centre for Quantum Geometry of Moduli Spaces) funded by the Danish National Research Foundation.}}

\vspace{.6truecm}

\centerline{$^1${\em Institut f\"ur Mathematik}}
\centerline{{\em Humboldt-Universit\"at zu Berlin}}
\centerline{{\em Rudower Chaussee 25, 12489 Berlin, Germany}}

\vspace{.3truecm}
\centerline{$^2${\em Center for Quantum Geometry of Moduli Spaces}}
\centerline{{\em Department of Mathematical Sciences, Aarhus University}}
\centerline{{\em Ny Munkegade 118, bldg. 1530, DK-8000 Aarhus C,
Denmark}}

\vspace{1.0truecm}

\begin{abstract}
In arXiv:1008.1018 it is shown that a given stable vector bundle $V$ on a Calabi-Yau threefold $X$ which satisfies $c_2(X)=c_2(V)$ can be deformed to a solution of the Strominger system and the equations of motion
of heterotic string theory. In this note we extend this result to the polystable case and construct explicit examples of polystable bundles on elliptically fibered Calabi-Yau threefolds where it applies. The polystable bundle is  given by a spectral cover bundle, for the visible sector, and a suitably chosen bundle, for the hidden sector. This provides a new class of heterotic flux compactifications via non-K\"ahler deformation of Calabi-Yau geometries with polystable bundles. As an application, we obtain examples of non-K\"ahler deformations of some three generation GUT models.
\end{abstract}

\newpage
\section{Introduction}
The geometry of the internal six-dimensional manifold of a supersymmetric compactification of the heterotic
string with non-zero torsion is characterized by a hermitian $(1,1)$-form $\omega$, a holomorphic $(3,0)$-form $\Omega$ and a non-abelian gauge field strength $F$. Supersymmetry and anomaly cancellation
constrain possible triples $(\omega, \Omega, F)$ which must satisfy
\begin{align}\label{stromsys}
F^{2,0}=F^{0,2}=0, \ \ \  F\wedge \omega^2&=0,\\
d(||\Omega||_\omega \omega^2) & = 0,\label{eq:epsilonconfbalanced}\\
i\partial\dbar\omega - \alpha'(\tr(R\wedge R)-\tr(F\wedge F))&=0\label{anom}.
\end{align}
This system of equations (usually called {\it the Strominger system}) provides the necessary
and sufficient conditions for $N=1$ space-time supersymmetry in four dimensions \cite{Strom}.
Moreover, in \cite{Strom} it is shown that given a solution of this system, the physical fields can be expressed in terms of $(\omega, \Omega, F)$.

First examples of solutions of the Strominger system have been obtained in \cite{Yau1, FuYau, FTY, LiYau}. One difficulty in obtaining smooth solutions of the Strominger system lies in the fact that many theorems of K\"ahler geometry and thus methods of algebraic geometry do not apply. Therefore one approach to obtain solutions is to simultaneously perturb K\"ahler and Hermitian-Yang-Mills (HYM) metrics and so avoid the direct construction of non-K\"ahler manifolds with stable bundles. This approach has been used in \cite{LiYau} where it is shown that a deformation of the holomorphic structure on the direct sum of the tangent bundle and the trivial bundle and of the K\"ahler form of a given Calabi-Yau threefold leads to a smooth solution of the Strominger system whereas the original Calabi-Yau space is perturbed to a non-K\"ahler space. For other constructions we refer to \cite{DaRa, GoPro, BeBe, BeSe}.

The aim of this note is to obtain solutions of the Strominger system using a perturbative method which has been inspired by the method developed in \cite{LiYau} and recently used to construct solutions in \cite{AGF}. Our starting point will be here a solution of the Strominger system in the large radius limit, given by any polystable vector bundle $V_1\oplus V_2$ on a Calabi-Yau threefold $X$ with $SU(3)$ holonomy which satisfies the second Chern class constraint
\begin{equation}\label{secondC}
c_2(X)=c_2(V_1)+c_2(V_2),
\end{equation}
where $V_2$ corresponds to a gauge bundle in the {\it hidden sector} of the $E_8\times E_8$ heterotic string theory.
This solution is then perturbed to a solution of the system using the implicit function theorem similarly as in
\cite{AGF}. Moreover, the solutions obtained satisfy also the equations of motion derived from the effective action of the heterotic string. For this we use a result of \cite{FIVU} and extend the original Strominger system as in \cite{AGF} by an instanton condition
\begin{equation}\label{eq:epsiloninstR}
R^{2,0}=R^{0,2}=0, \ \ \ R\wedge \omega^2=0.
\end{equation}
This condition is obtained in~\cite{FIVU} considering the equations of motion of the heterotic string up to two-loops in sigma model perturbation theory as exact. Now having obtained a supersymmetric solution to 
the low energy field theory approximation of heterotic string theory, one can argue that such solutions can be completed to solutions to all orders of the fully $\alpha'$-corrected equations (cf. \cite{Wi1, WiWi, WuWi}) using non-renormalization theorems for the effective four-dimensional superpotential.

In this note we construct examples where the deformation result applies, given by a class of polystable bundles on an elliptically fibered Calabi-Yau threefold with
$V_1$ given by a particular spectral cover bundle and $V_2$ given by a suitably chosen pullback bundle.
Note that due to the appropriate choice of the hidden sector bundle $V_2$ we avoid the otherwise necessary introduction of background five-branes which would lead to a further modification of the anomaly constraint
(for a discussion on this see \cite{AGF}). For another class of hidden sector bundles we obtain
examples of non-K\"ahler deformations of some three generation GUT models.

\section{Solutions via Polystable Bundles}

In this section we prove an existence result for the Strominger system and the equations of motion in the polystable case, adapting a previous result in \cite{AGF}. As in~\cite{FIVU}, we consider the equations of motion of the heterotic string up to two-loops in sigma model perturbation theory as exact. We use the standard terminology in the mathematical literature and, for simplicity, we restrict to the case in which the Calabi-Yau threefold has strict $SU(3)$ holonomy. 

Let $X$ be a compact complex threefold with nowhere vanishing $(3,0)$-holomorphic form $\Omega$. Let $V$ be a holomorphic vector bundle over $X$ satisfying $c_2(V) = c_2(X)$ with rank $r$ and $c_1(V) = 0$. Consider the coupled system given by~\eqref{stromsys},~\eqref{eq:epsilonconfbalanced},~\eqref{anom} and~\eqref{eq:epsiloninstR} where $F$ is the curvature of a (unitary) connection $A$ on $V$, $\omega$ is a hermitian form on $X$ and $R$ is the curvature of a $\omega$-unitary connection $\nabla$ on $TX$.

Suppose that $X$ admits a Calabi-Yau metric $\omega_0$ with strict $SU(3)$ holonomy and irreducible Chern connection $\nabla_0$ and that $V$ is polystable with respect to $[\omega_0]$. Then,
\begin{equation}\label{eq:stables}
V = \oplus_{j=1}^k V_j
\end{equation}
is a direct sum of stable bundles with respect to $[\omega_0]$. We claim that there exists a family of solutions $(A_\lambda,\omega_\lambda,\nabla_\lambda)$ of~\eqref{stromsys},~\eqref{eq:epsilonconfbalanced},~\eqref{anom} and~\eqref{eq:epsiloninstR} parameterized by $\lambda \in ]\lambda_0,+\infty[$ such that $(A_\lambda,\omega_\lambda/\lambda,\nabla_\lambda)$ converges to a HYM connection on $V$, the Calabi-Yau metric $\omega_0$ and its Chern connection $\nabla_0$, respectively. In order to prove this result, following~\cite{AGF} we rescale the hermitian form $\omega$ by a positive constant $\lambda > 0$ and consider a new system of equations, the $\epsilon$-system, defined by the change of variable $\alpha' \to \epsilon = \alpha'/\lambda$ in~\eqref{anom}. Therefore, any solution of the $\epsilon$-system with $\epsilon > 0$ is related to the original system after rescaling. In the limit $\lambda\to\infty$ a solution of the $\epsilon$-system is given by a degree zero polystable holomorphic vector bundle on a Calabi-Yau threefold. This way the pair $(X,V)$ is regarded as a solution of the initial system in the large radius limit.

We prove now that $(X,V)$, the given solution for $\epsilon=0$, can be perturbed to a solution with small $\epsilon > 0$, i.e., with large $\lambda$. More precisely, the K\"ahler form of the given Calabi-Yau threefold is perturbed to a conformally balanced Hermitian form on the fixed complex manifold while also perturbing its Chern connection and the unique HYM connection on the bundle $V$, whereas we preserve the HYM condition. The argument is analogue to that in~\cite{AGF}, so we give a brief sketch here and refer to~\cite{AGF} for more details. We consider a $1$-parameter family of maps between (suitable) fixed Hilbert spaces
$$
\textbf{L}^\epsilon\colon \mathcal{V}_1 \to \mathcal{V}_2,
$$
indexed by $\epsilon \in \mathbb{R}$ which acts on triples $(A,\omega,\nabla) \in \mathcal{V}_1$ and whose zero locus corresponds to solutions of the $\epsilon$-system. Roughly speaking, $\textbf{L}^\epsilon$ is defined by sending a triple $(A,\omega,\nabla)$ to the $4$-tuple given by the LHS of~\eqref{stromsys},~\eqref{eq:epsilonconfbalanced},~\eqref{anom} (with $\alpha' \to \epsilon = \alpha'/\lambda$) and~\eqref{eq:epsiloninstR}. This way, the map $\textbf{L}^\epsilon$ splits into a direct sum
$$
\textbf{L}^\epsilon = \textbf{L}_1\oplus\textbf{L}_2^\epsilon\oplus\textbf{L}_3\oplus\textbf{L}_4
$$
and the equation $\textbf{L}^\epsilon = 0$ with $\epsilon \neq 0$ corresponding to the $\epsilon$-system is related to a simpler one at $\epsilon = 0$ which has a known solution $(A_0,\omega_0, \nabla_0)$. As the initial bundle $V$ is polystable we start with a reducible connection $A_0 = \oplus^j A_j$, with irreducible HYM pieces $A_j$ corresponding to the direct sum decomposition~\eqref{eq:stables} into stable subbundles. Then, for the perturbative argument we consider connections $A$ which are direct sum of connections on the subbundles $E_j$ (for the definition of the Hilbert spaces $\mathcal{V}_1$, $\mathcal{V}_2$ above) and split the map $\textbf{L}_1 = \oplus^j\textbf{L}_{1,j}$ in the obvious way. Then, the differential
$$
\delta_0 \textbf{L}^0\colon \mathcal{V}_1 \to \mathcal{V}_2,
$$
at the initial solution is an isomorphism (after composition with a suitable orthogonal projection) and an implicit function theorem argument shows that there is a solution of $\textbf{L}^\epsilon = 0$ nearby $(A_0,\omega_0, \nabla_0) \in \mathcal{V}_1$ in an open neighborhood of $\epsilon = 0 \in \mathbb{R}$. Therefore, we obtain the following result.

\begin{thm}\label{thm:existence0}
Let $V$ be a degree zero holomorphic vector bundle over a compact Calabi-Yau threefold $(X,\omega_0)$ with holonomy equals to $SU(3)$. If $c_2(V) = c_2(X)$ and $V$ is polystable with respect to $[\omega_0]$, then there exists $\lambda_0 \gg 0$ and a $C^1$ curve
$$
]\lambda_0,+\infty[ \ni \lambda \to (A_\lambda,\omega_\lambda,\nabla_\lambda)
$$
of solutions of the Strominger system and the equations of motion derived from the effective heterotic string action such that $\dbar_\lambda$ is isomorphic to $\dbar_0$ for all $\lambda$ and $(A_\lambda,\frac{\omega_\lambda}{\lambda},\nabla_\lambda)$ converges uniformly to the solution $(A_0,\omega_0,\nabla_0)$ when $\lambda \to \infty$.
\end{thm}

The perturbation provided by Theorem~\ref{thm:existence0} leaves the holomorphic structure of $V$ unchanged while the one on $TX$ is shifted by a complex gauge transformation and so remains isomorphic to the initial one. Since the initial bundle is polystable the solutions of the hermitian-Yang-Mills equations obtained in the perturbation may be reducible.

\section{Application to Elliptic Calabi-Yau Threefolds}

One motivation for this note was to apply the method developed in \cite{AGF}
to the widely explored class of spectral cover bundles which has been originally constructed in \cite{FMW1, FMW2}.
As shown in \cite{Pantev}, these bundles are not capable to solve $c_2(X)=c_2(V)$ but do solve the generalized
anomaly constraint $c_2(X)-c_2(V)=[W]$
with $[W]$ an effective curve class \cite{FMW1}. Here $[W]\neq 0$ corresponds to the presence of background five-branes in the heterotic vacuum. One possibility to avoid background five-branes (as these cause a delta function contribution to the anomaly equation, cf. discussion in \cite{AGF}) is to search for stable bundles whose second Chern class is equal to $[W]$ corresponding
to turning on a {\it hidden sector} bundle which can then be used to break part of the hidden $E_8$ gauge symmetry.
(The general question to give sufficient conditions for the existence of stable bundles with prescribed Chern classes has been studied in \cite{DRY} and \cite{AC10}).

Following this approach, in this section we will construct a class of polystable bundles $V\oplus \pi^*E$ where
$V$ is given by a spectral cover bundle and $\pi^*E$ a pullback bundle. For this we first briefly
recall the spectral cover construction of \cite{FMW1} and then prove that $\pi^*E$
can be chosen to be stable with respect to the same K\"ahler class $\omega$ as $V$, that is, $\omega=t\omega_0+\pi^*\omega_B$ (cf.below). We then provide conditions
such that $V\oplus \pi^*E$ satisfies the second Chern class constraint (\ref{secondC}) imposed by anomaly cancellation.

On elliptically fibered Calabi-Yau threefolds $\pi\colon X\to B$ with section $\sigma$ a general construction of stable vector bundles
was given in \cite{FMW1}. We will consider here the case of stable $SU(n)$ bundles $V$. Such bundles
can be explicitly constructed from a divisor $C$ in $X$, called the spectral cover, and a line bundle $L$
on $C$. The spectral cover $C$ is a surface in $X$ that is an $n$-fold cover of the base $B$. More precisely,
$C$ is an element of the linear system $|n\sigma+\pi^*\eta|$ of ${\mathcal O}_X(n\sigma+\pi^*\eta)$ whereas $\eta$ is a curve class in $B$. Let us briefly recall the conditions on the spectral data $(C,L)$ such that $V$ is a stable $SU(n)$ bundle. For more details see \cite{FMW1}.

For $SU(n)$ bundles the first Chern class of $L$ is given by \cite{FMW1}
\begin{equation}
c_1(L)=n(\frac{1}{2}+\mu)\sigma+(\frac{1}{2}-\mu)\pi^*\eta+(\frac{1}{2}+n\mu)\pi^*c_1(B)
\end{equation}
since $c_1(L)$ must be an integer class, it follows that: if $n$ is odd, then $\mu$ is strictly half
integral and for $n$ even an integral $\mu$ requires
$\eta\equiv c_1 \  ({\rm mod}\  2)$ while a strictly
half-integral $\mu$ requires $c_1(B)$ even.

Now given the spectral data $(C,L)$, there is an unique $SU(n)$ bundle $V$
which can be constructed canonically in two ways. As in \cite{FMW1}, $V$ can be directly constructed as
$$V=\pi_{1*}(\pi_2^*L\otimes {\mathcal P})$$
where ${\mathcal P}$ is the associated Poincar\'e bundle and $\pi_i$ are the two projections of the fiber
product $X\times_B C$ onto the two factors $X$ and $C$. Equivalently, one can construct $V$ directly
from the spectral data using a relative Fourier-Mukai transform (for a survey see \cite{AHR}).

In Theorem 7.1 in \cite{FMW2} it is shown that for irreducible $C$ the bundle $V$ is stable with respect to the K\"ahler class $\omega=t \omega_0+\pi^*\omega_B$ with $\omega_0$ some fixed K\"ahler class in $X$ and $\omega_B$ some K\"ahler class in $B$ and $t$ a sufficiently small positive real number. To assure that the linear system $|n\sigma+\pi^*\eta|$ contains an irreducible surface $C$ one has to demand that the linear system $|\eta|$ be base-point free in $B$ and that
$\eta-nc_1(B)$ is effective in $B$ \cite{FMW1}.

The second Chern classes of $X$ and $V$ are given by \cite{FMW1}
\begin{eqnarray}
c_2(X)&=&12\sigma\cdot \pi^*c_1(B)+\pi^*(c_2(B)+11c_1(B)^2)\nonumber\\
c_2(V)&=&\pi^*\eta\cdot \sigma-\frac{n^3-n}{24}\pi^*c_1(B)^2+\frac{1}{2}(\mu^2-\frac{1}{4})n\pi^*\eta\cdot(\pi^*\eta-n\pi^*c_1(B)).\nonumber
\end{eqnarray}

As we are interested in solving the topological constraint with spectral cover bundles
and $[W]=0$ (cf. above) we need to find another class of $\omega$-stable bundles. For this
let us consider pullback bundles which have been also studied in \cite{ACext}. However,
in this note we will not assume that the base of the elliptically fibered Calabi-Yau threefold has
an ample anti-canonical line bundle as it was assumed in the proof of Lemma 5.1 in \cite{ACext}.
Note that in \cite{ACext} the loss of generality in $B$ allowed to freely specify $t$ in $\omega=t\omega_0+\omega_B$ which was required to construct stable extension bundles in a second step.

Let $X$ be an elliptic Calabi-Yau threefold $\pi\colon X\to B$ with section $\sigma$ and let $E$ be
a given $\omega_B$-stable rank $r$ vector bundle on $B$ (a rational surface) with $c_1(E)=0$ and $c_2(E)=k$. We get the following result:

\begin{thm}
$\pi^*E$ is stable on $X$ with respect to $\omega=t\omega_0+\pi^*\omega_B$ for $t$ sufficiently small.
\end{thm}

\begin{proof} Let $F$ be a subsheaf of $\pi^*E$ with $0<{\rm rank}\ F<r$. We need to show that $\mu(F)<0$. From Lemma 4.1 in \cite{ACext} we get for a subsheaf $F$ of $\pi^*E$ that $c_1(F)\cdot \pi^*\omega_B^2\leq 0$ and $c_1(F)\cdot \omega_0\cdot \pi^*\omega_B<0$. Now using Lemma 7.2 in \cite{FMW2} we have for all nonzero subsheaves $F$ of $\pi^*E $ that $\frac{c_1(F)\cdot\pi^*\omega_B^i\cdot \omega_0^{2-i}}{{\rm rank}\ F}\leq A$ with $0\leq i\leq 2$ and $A$ only depending on $\pi^*E$, $\omega_0$ and $\pi^*\omega_B$. We can thus conclude that for $t$ chosen sufficiently small $\mu(F)<0$. 
\end{proof}

So far we have constructed a class of polystable vector bundles $V\oplus\pi^*E$ with $V$ a given
spectral cover bundle. Now if we impose the following conditions (set $c_i(B)=c_i$)
\begin{eqnarray}
\eta&=&12c_1\\
k&=&c_2+11c_1^2+\frac{n^3-n}{24}c_1^2-(\mu^2-\frac{1}{4})\frac{n}{2}\eta(\eta-nc_1)\label{insta}
\end{eqnarray}
then $c_2(X)=c_2(V)+c_2(\pi^*E)$ is satisfied and thus gives examples
of polystable bundles which satisfy the hypothesis of Theorem 2.1.

It remains to prove that a $\omega_B$-stable rank $r$ bundle $E$ on $B$ with
$c_1(E)=0$ and $c_2(E)=k$ and $k$ given by (\ref{insta}) actually exists on $B$. In \cite{Art91} it is shown that the moduli space of vector bundles of rank $r$, $c_1(E)=0$ and $c_2(E)=k$ is not empty
if $k>{\rm max}(1,p_g)(r+1)$ (where $p_g$ denotes the geometric genus of $B$). Thus in our case we have to assure that (using $\eta=12c_1$ and $c_1^2+c_2=12$ for $B$ rational)
\begin{equation}\label{inst}
12+\Big(10+\frac{n^3-n}{24}-(\mu^2-\frac{1}{4})(72-6n)n\Big)c_1^2>r+1.
\end{equation}
An example is given by spectral bundles with $n$ odd and $\mu=\pm\frac{1}{2}$ or with $n$ even, $\mu=0$ and $c_1$ even
solve this condition ($r\leq 8$ is assumed, for breaking the hidden $E_8$ gauge symmetry).

\section{Towards GUT Model Building}

The heterotic string provides a natural setting for building realistic models of particle physics.
In recent years much progress has been made in the construction of physically interesting
Calabi-Yau compactifications with stable vector bundles which lead to GUT models and models with
a particle content close to the Standard Model (for recent progress see, \cite{AnCu07}, \cite{BouDo06}, \cite{BraHeOv}). However, it is an open problem
whether these phenomenologically interesting models can be generalized to torsional heterotic backgrounds.
If this were possible, then it would allow to address the question of solving the anomaly equation and not only
its integrability condition (which is necessary but not sufficient) as it is done in most of the present constructions,
and it would also allow to address the question of moduli stabilization.

In this section we will study this problem for some three generation GUT models. We will
construct a polystable bundle which leads to three generation models of unbroken gauge group
$SU(5)$ or $E_6$ in the visible sector and which then can be perturbed using the result of Theorem 2.1
to a solution of the Strominger system and the equation of motion (in lowest order of $\alpha'$, cf. above).
Therefore the GUT models qualify as torsional heterotic backgrounds.

The net number of generations $N_{gen}$ is determined by the index of the Dirac operator coupled to the
visible sector gauge bundle $V_{vis}$ and is given by the third Chern class of $V_{vis}$
$$N_{gen}=\frac{1}{2}c_3(V_{vis}).$$
For spectral cover bundles one finds (cf. \cite{And, Cur, Dia}) that $N_{gen}$ is given by
$$N_{gen}=\mu\pi^*\eta(\pi^*\eta-n\pi^*c_1(B))\sigma.$$
Now obviously the above examples with $\eta=12c_1$ (and say $\mu=\pm \frac{1}{2}$) lead to rather large
generation numbers. To get more flexibility for GUT model building we will consider now another possibility.
As above, we will construct a polystable bundle
$$V=V_{vis}\oplus V_{hid},$$
with $V_{vis}$ given by a spectral cover $SU(n_1)$ bundle. For the hidden bundle we take the polystable bundle
$$V_{hid}=W\oplus \pi^*E,$$
with $W$ an $SU(n_2)$ spectral cover bundle and $\pi^*E$ an $SU(r)$ $\omega_B$-stable pullback bundle.
Alternatively, one could consider a stable deformation of $W\oplus\pi^*E$ which would then lead to 
an irreducible structure group of the hidden bundle and thus give different unbroken hidden gauge groups. 
Necessary and sufficient conditions for the existence of such a deformation have been obtained for polystable bundles in \cite{Huy}. 

The corresponding Chern classes are given by $c_2(\pi^*E)=k$ and
\begin{eqnarray}
c_2(V_{vis})&=&\pi^*\eta_1\cdot \sigma-\frac{n_1^3-n_1}{24}\pi^*c_1(B)^2+\frac{1}{2}(\mu_1^2-\frac{1}{4})n_1\pi^*\eta_1\cdot(\pi^*\eta_1-n_1\pi^*c_1(B))\nonumber\\
c_2(W)&=&\pi^*\eta_2\cdot \sigma-\frac{n_2^3-n_2}{24}\pi^*c_1(B)^2+\frac{1}{2}(\mu_2^2-\frac{1}{4})n_2\pi^*\eta_2\cdot(\pi^*\eta_2-n_2\pi^*c_1(B)).\nonumber
\end{eqnarray}
Now similar as above, if we impose the conditions
\begin{eqnarray}
\eta_1+\eta_2&=&12c_1 \label{wb}\\
k&=&c_2+11c_1^2+\sum_{i=1}^2\Big[\frac{n_i^3-n_i}{24}c_1^2-(\mu_i^2-\frac{1}{4})\frac{n_i}{2}\eta_i(\eta_i-n_ic_1)\Big]
\end{eqnarray}
then $c_2(X)=c_2(V_{vis})+c_2(V_{hid})$ is satisfied and thus gives examples
of polystable bundles which satisfy the hypothesis of Theorem 2.1. Moreover, as above in (\ref{inst}) we have
here also to assure that $k>r+1$ which guarantees the existence of stable rank $r$ bundles $E$ on $B$ with the Chern classes $c_1(E)=0$ and $c_2(E)=k$.

For concreteness, we will now show that the construction leads to three generation models of unbroken
gauge groups $E_6$ and $SU(5)$ in the visible sector, that is, we specify $SU(n_1)$ bundles $V_{vis}$ with $n_1=3,5$. For simplicity, we will fix the base $B$ to be given by
the Hirzebruch surface ${\mathbb F}_1$ with Chern classes $c_1({\mathbb F}_1)=2b+3f, c_2({\mathbb F}_1)=4$ (and with the intersection relations $b^2=-1$, $bf=1$ and $f^2=0$). Moreover, we will assume that $r\leq 8$.
\begin{ex}
We obtain an unbroken $E_6$ gauge group for $n_1=3$ and $\eta_1=11b+10f$ and $\mu_1=\pm \frac{1}{2}$ giving three generations. Therefore by (\ref{wb}) we get $\eta_2=13b+26f$ and the instanton number of $E$ is then given by
\begin{equation}
k=100+\frac{n_2^3-n_2}{24}c_1^2-(\mu_2^2-\frac{1}{4})\frac{n_2}{2}\eta_2(\eta_2-n_2c_1)
\end{equation}
thus if we take $n_2=3$ and $\mu_2=\pm\frac{1}{2}$ giving $k=108$ and all conditions are satisfied.
\end{ex}
\begin{ex}
An unbroken $SU(5)$ gauge group in the visible sector is obtained for $n_1=5$ and $\eta_1=12b+15f$ and $\mu_1=\pm \frac{1}{2}$ giving three generations. By (\ref{wb}) we get $\eta_2=12b+21f$ and the instanton number of $E$ is then given by
\begin{equation}
k=132+\frac{n_2^3-n_2}{24}c_1^2-(\mu_2^2-\frac{1}{4})\frac{n_2}{2}\eta_2(\eta_2-n_2c_1)
\end{equation}
as $\eta_2\equiv c_1 \ ({\rm mod}\ 2)$ we can take here $n_2=2$ and $\mu_2=0$ giving $k=197$ which guarantees
the existence of the required rank two bundle $E$ on $B$.
\end{ex}
\vskip 1cm
We thank the SFB 647 ``Space-Time-Matter, Analytic and Geometric Structures'' and the Humboldt-University Berlin for financial support and the Free-University Berlin for hospitality. M.G.F. thanks the Max-Planck-Institute for Mathematics and the Centre for Quantum Geometry of Moduli Spaces for financial support.


\begin{thebibliography}{OG1}

\bibitem{AGF} B. Andreas and M. Garcia-Fernandez, {\em Solutions of the Strominger system via stable bundles on Calabi-Yau threefolds}, arXiv:1008.1018 [math.DG] (2010).

\bibitem{AC10} B. Andreas and G. Curio, {\em On possible Chern classes of stable bundles on Calabi-Yau threefolds}, 1010.1644 [hep-th] (2010).

\bibitem{ACext} B. Andreas and G. Curio, {\em Stable bundle extensions on elliptic
Calabi-Yau threefolds}, J. Geom. Phys. {\bf 57}, 2249-2262, 2007, math.AG/0611762.

\bibitem{AHR} B. Andreas and D. Hernandez-Ruiperez, {\em Fourier-Mukai transforms and applications to string theory}, Rev. R. Acad. Cienc. Exactas Fis. Nat. Ser. A Mat, {\bf 99} (1) (2005) 29-77, math.AG/0412328.

\bibitem{And} B. Andreas, {\em On vector bundles and chiral matter in $N=1$ heterotic compactifications},
JHEP 9901:011,1999, hep-th/9802202.

\bibitem{AnCu07} B. Andreas and G. Curio, {\em Deformations of bundles and the standard model}, Phys. Lett. {\bf B655} (2007) 290-293, 0706.1158 [hep-th].

\bibitem{Art91} I. V. Artamkin, {\em Stable bundles with $c_1=0$ on rational surfaces}, Math. USSR Izv. {\bf 36}, 231-246 (1991).

\bibitem{BeBe} K. Becker, M. Becker, J. X. Fu, L. S. Tseng and S.-T. Yau, {\em Anomaly cancellation
and smooth non-K\"ahler solutions in heterotic string theory}, Nucl. Phys. {\bf B751} (2006) 108, hep-th/0604137.

\bibitem{BeSe} K. Becker and S. Sethi, {\em Torsional heterotic geometries}, 0903.3769 [hep-th].

\bibitem{BouDo06} V. Bouchard and R. Donagi, {\em An $SU(5)$ heterotic standard model}, Phys. Lett. {\bf B633} (2006) 783-791, hep-th/0512149.

\bibitem{BraHeOv} V.Braun, Y.-H. He, B. A. Ovrut and T. Pantev, {\em The exact MSSM spectrum from string theory}, JHEP 05 (2006) 043, hep-th/0512177.

\bibitem{Cur} G. Curio, {\em Chiral matter and transitions in heterotic string models}, Phys.Lett. {\bf B435} 39-48,1998, hep-th/9803223.

\bibitem{DaRa} K. Dasgupta, G. Rajesh and S. Sethi, {\em M-theory, orientifolds and G-flux}, JHEP 08:023 (1999), hep-th/9908088.

\bibitem{Dia} D.~E. Diaconescu, G. Ionesei, {\em Spectral covers, charged matter and bundle cohomology},
HEP 9812:001,1998, hep-th/9811129.

\bibitem{DRY} M.R.~Douglas, R.~Reinbacher and S.-T.~Yau, {\em Branes, Bundles and Attractors: Bogomolov and Beyond}, math.AG/0604597.

\bibitem{FIVU} M. Fern\'andez, S. Ivanov, L. Ugarte, R. Villacampa, {\em Non-K\"aehler heterotic-string compactifications with non-zero fluxes and constant dilaton}, Commun. Math. Phys. {288} (2009) 677-697, 0804.1648 [hep-th].

\bibitem{FMW1} R. Friedman, J. Morgan and E. Witten, {\em Vector bundles and F-theory},
Comm. Math. Phys. {\bf 187} (1997) 679, hep-th/9701162.

\bibitem{FMW2} R. Friedman, J.W. Morgan and E. Witten, {\em Vector bundles over elliptic fibrations}, Jour. Alg. Geom. {\bf 8} (1999) 279-401, alg-geom/9709029.

\bibitem{FuYau} J.-X. Fu and S.-T.~Yau, {\em The theory of superstring with flux on non-K\"ahler manifolds and the complex Monge-Amp\`ere}, J. Diff. Geom. 78 (2008) 369--428, 0604063 [hep-th].

\bibitem{FTY} J.-X.~Fu, L.-S.~Tseng and S.-T.~Yau, {\em Local heterotic torsional models}, Comm. Math. Phys. 289 (2009), 1151--1169, 0806.2392 [hep-th].

\bibitem{Yau1} J.~Fu, J.~Li and S.-T.~Yau, {\em Constructing balanced metrics on some families of non-K\"ahler Calabi-Yau threefolds}, 0809.4748 [math.DG].

\bibitem{GoPro} E. Goldstein and S. Prokushkin, {\em Geometric model for complex non-K\"ahler manifolds with
$su(3)$ structure}, Comm. Math. Phys. {\bf 251} (2004) 65-78, hep-th/0212307.

\bibitem{Huy} D. Huybrechts, {\em The tangent bundle of a Calabi-Yau manifold - deformations and restriction
to rational curves}, Comm. Math. Phys. {\bf 171} (1995) 139-158.

\bibitem{LiYau} J.~Li and S.-T.~Yau, {\em The existence of supersymmetric string theory with torsion}, J. Diff. Geom. 70 (2005) 143-181, 0411136 [hep-th].

\bibitem{Pantev} B. A. Ovrut, T. Pantev and J. Park, {\em Small instanton transition in heterotic M-theory},
JHEP 0005: 045, 2000, hep-th/0001133.

\bibitem{Strom} A.~Strominger, {\em Superstrings with torsion}, Nucl. Phys. B {274} (1986) 253-284.

\bibitem{Wi1} E.~Witten, {\em New Issues in manifolds of SU(3) holonomy}, Nucl. Phys. B {268} (1986) 79-112.

\bibitem{WiWi} E.~Witten and L.~Witten, {\em Large radius expansion of superstring compactification}, Nucl. Phys. B {281} (1987) 109-126.

\bibitem{WuWi} X.~Wu and L.~Witten, {\em Space-time supersymmetry in large radius expansion of superstring compactification}, Nucl. Phys. B {289} (1987) 385-396.


\end{thebibliography}
\end{document}